\documentclass[10pt,a4paper]{article}

\usepackage{amsmath,amssymb,amsthm,mathrsfs}
\usepackage{stmaryrd} 
\usepackage{graphicx}
\usepackage{xcolor}
\usepackage{soul}
\usepackage{cancel}
\usepackage{changes}
\usepackage{fullpage}

\newlength{\defbaselineskip}
\newcommand{\setlinespacing}[1]%
{\setlength{\baselineskip}{#1 \defbaselineskip}}

\theoremstyle{plain}
\newtheorem{theorem}{Theorem}[section]
\newtheorem{lemma}[theorem]{Lemma}
\newtheorem{proposition}[theorem]{Proposition}
\newtheorem{corollary}[theorem]{Corollary}

\theoremstyle{definition}
\newtheorem{definition}[theorem]{Definition}

\newtheorem{ass}[theorem]{Assumption}

\theoremstyle{remark}
\newtheorem{remark}[theorem]{Remark}

\numberwithin{equation}{section}

\DeclareMathOperator*{\esssup}{ess\,sup}

\newcommand{\cM}{\mathcal{M}}

\newcommand{\bG}{\mathbb{G}}

\newcommand{\stkout}[1]{\ifmmode\text{\sout{\ensuremath{#1}}}\else\sout{#1}\fi}
\newcommand{\E}{{\mathbb E}}
\newcommand{\R}{{\mathbb R}}
\newcommand{\cE}{\ensuremath{\mathcal{E}}}

\newcommand{\dd}{\ensuremath{\operatorname{d}\! }}
\newcommand{\dt}{\ensuremath{\operatorname{d}\! t}}
\newcommand{\de}{\ensuremath{\operatorname{d}\! e}}
\newcommand{\ds}{\ensuremath{\operatorname{d}\! s}}
\newcommand{\dr}{\ensuremath{\operatorname{d}\! r}}

\newcommand{\dw}{\ensuremath{\operatorname{d}\! W}}

\usepackage[colorlinks=true, citecolor=blue]{hyperref}

\allowdisplaybreaks
\begin{document}
\title{A System of BSDEs with Singular Terminal Values Arising in Optimal Liquidation with Regime Switching}

\author{Guanxing Fu\footnote{Department of Applied Mathematics, The Hong Kong Polytechnic University, Kowloon, Hong Kong, China; \url{guanxing.fu@polyu.edu.hk}} \and Xiaomin Shi\footnote{School of Statistics and Mathematics, Shandong University of Finance and Economics, Jinan 250100, China; \url{shixm@mail.sdu.edu.cn}} \and Zuo Quan Xu\footnote{Department of Applied Mathematics, The Hong Kong Polytechnic University, Kowloon, Hong Kong, China; \url{maxu@polyu.edu.hk}}
}

\maketitle

\begin{abstract}
We investigate a stochastic control problem with regime switching, which arises in the context of an optimal liquidation problem involving dark pools and multiple regimes. A novel aspect of this model is the introduction of a system of backward stochastic differential equations with jumps and singular terminal values, appearing for the first time in the literature. We establish the existence result for this system, thereby solving the stochastic control problem with regime switching. More importantly, we also achieve a uniqueness result for this system, which contrasts with the minimal solutions typically found in most related literature.
\end{abstract}

{\bf AMS Subject Classification:} 93E20, 91B70, 60H30

{\bf Keywords:}{ Optimal liquidation, system of BSDEs with singular terminal values, regime switching }

\section{Introduction and overview}
Consider the following stochastic control problem with regime switching:
\begin{align}\label{cost}
	J(\xi,\beta;x_0,i_0)=\E\left[\int_0^T\left(\eta_t^{\alpha_{t-}}\xi_t^2+\lambda_t^{\alpha_{t-}}X_t^2
	+\int_{\cE}\gamma_t^{\alpha_{t-}}(e)|\beta_t(e)|^2\nu(\de)\right)\dt\right]\rightarrow \min_{\xi,\beta} ,
\end{align}
where the dynamics $X$, controlled by $(\xi,\beta)$, satisfies
\begin{align}\label{state}
	X_t=x_0 - \int_0^t\xi_s\ds - \int_0^t\int_{\cE}\beta_s(e)N(\ds,\de), \quad X_T=0.
\end{align}
Here, $N$ is a Poisson random measure with intensity measure $\nu$, $\alpha$ is a Markov chain with $\ell$ states and $i_0$ is the regime before the first switching of $\alpha$. {Without loss of generality (w.l.o.g.), let $x_0$ be a positive constant.}

The special case of \eqref{cost}-\eqref{state}, where $N=\nu=0$, $\ell=1$ and $\eta$ and $\lambda$ are positive constants, corresponds to the classical optimal liquidation model of Almgren and Chriss \cite{AC-2001}. Specifically, let $X_t$ denote the position of the investor at time $t\in[0,T]$, who would like to close her initial position $X_0=x_0$ with $x_0>0$ for selling, subject to the liquidation constraint $X_T=0$, by submitting market order $\xi$ to the traditional exchange. Due to the limited liquidity, the market orders, which are visible to the public, will typically add a temporary impact $\eta\xi$ to the benchmark price and result in a quadratic trading cost $\eta\xi^2$. The second term in \eqref{cost} is considered as a risk aversion, penalizing slow liquidation. 

In contrast to the market orders, the passive orders submitted to dark pools are hidden from the public view and thus results in no price impact, however, its execution is uncertain. We use a Poisson process to model the execution time and $\beta$ denotes the passive order. Instead of moving the price in an unfavorable direction, passive orders will lead to an adverse selection cost denoted by $\gamma\beta^2$; see e.g. \cite{HN-2014, Kratz-2014,Kratz-2018}.

The main characteristic of optimal liquidation problem is a singular terminal condition of the value function, characterized by ODE, PDE or BS(P)DE depending on the model assumptions. Two main techniques have been developed to address the singularity. The first approach is to penalize the final open position and transform the backward equations with singular terminal values into ones with bounded terminal values, and then send the degree of penalization to infinity; refer to \cite{AJK-2014,GHQ-2015,KP-2016,PZ-2019}. {Among them, motivated by portfolio liquidation problems, Ankirchner et al \cite{AJK-2014} solve a BSDE with singular terminal values and integrable coefficients, Graewe et al \cite{GHQ-2015} consider a non-Markovian liquidation problem by studying a BSPDE with singular terminal values, {Kruse and Popier \cite{KP-2016}} study a BSDE with jumps and singular terminal values in a general filtered probability space, arising in a liquidation problem with dark pools, and Popier and Zhou \cite{PZ-2019} study an optimal liquidation problem with uncertainty by studying a second-order BSDEs with singular terminal values}. The second approach investigates the asymptotic behavior of the solution around the maturity and transform the backward equations with singular terminal vlaues into ones with bounded terminal values but singular generators;
refer to \cite{GH-2017,GHS-2018,GP-2021}, {where Graewe and Horst \cite{GH-2017} study a three-dimensional fully coupled BSDE with one component having singular terminal values arising in a liquidation problem with transient impact, Graewe et al \cite{GHS-2018} solve a PDE with singular terminal value in a Markovian and price-sensitive liquidation model, and Graewe and Popier \cite{GP-2021} establish the wellposedness result for a very general class of BSDEs with singular terminal values. }

We will study the stochastic control problem \eqref{cost}-\eqref{state} in a non-Markovian framework and characterize the value function and optimal control by a system of BSDEs with jumps and singular terminal values.
In constrast to all aforementioned works, we assume that all parameters depend on the Markov chain $\alpha$. Our explanation in mind is to capture the regime shift of the market. For example, $\eta^\alpha$ captures the illiquidity level and $\lambda^\alpha$ is proportional to the price volatility; we use $\alpha$ to model the regimes of liquidity and volatility levels: low, medium and high if e.g. $\ell=3$.
The introduction of regime switching increases the dimensionality and leads to a system of coupled BSDEs with singular terminal values, which appears in the literature for the very first time; the multi-dimensional aspect in \cite{HX-2019} arises from the model's use of multiple assets, resulting in a matrix-valued BSDE with singular terminal values that differs from our system of coupled BSDEs.
The contribution of our paper is two-fold: existence and uniqueness results of the BSDE system. To obtain the existence result, we use the penalization method by truncating the singular terminal value with a finite one. To achieve a limit as the penalization approaches infinity, the challenge lies in establishing appropriate upper and lower bounds for the solution to the truncated the BSDE system. To address this, we develop a tailored comparison theorem for multidimensional BSDEs with monotone drivers. The key insight is that we can compare the solution to our the BSDE system with the solutions to two single ODEs, which serve as the upper and lower bounds according to our comparison theorem. Our existence result can be viewed as a multidimensional extension of \cite{KP-2016}, where Kruse and Popier only considered one-dimensional problem but in a very general setting. However, the solution established in \cite{KP-2016} is only shown to be minimal. Thus, compared to the existence result, our uniqueness result is more significant in the literature on BSDEs with singular terminal values.
Currently, most existing results concentrate on the minimal solution; see, for example, \cite{AJK-2014,KP-2016,Popier-2006,Popier-2007,SKP-2019} for BSDEs with singular terminal values, both with and without jumps. To the best of our knowledge, the uniqueness result for BSDEs with singular terminal values has only been addressed in \cite{GH-2017,GHS-2018,GP-2021,HX-2021}, all of which do not include jumps. Thus, even without regime switching, our uniqueness result is new in literature. To establish the uniqueness result, our approach is different from that in \cite{GP-2021}, which uses the equivalence between BSDEs with singular values and those with singular generators. Instead, our method is inspired by \cite{GHS-2018,HX-2021}. We first verify that the solution obtained in the existence result is minimal within a certain class of solutions. Then, we demonstrate that any solution to the system coincides with this minimal solution by showing that the value function of the optimal control problem \eqref{cost}-\eqref{state} can be characterized by the minimal solution. A key aspect of our proof is to establish an appropriate a priori estimate for any solution to the system of BSDEs {\it with singular terminal values}. Our approach heavily relies on our multidimensional comparison theorem and the liquidation constraint $X_T = 0$.

Before delving into the main text of the paper, we summarize results on stochastic control with regime switching in optimal liquidation problems, which or at least the conception are already explored in the literature; see \cite{Bian-2016,Pemy-2006,Pemy-2008,Elliott-2019}. Among these,
Pemy and Zhang \cite{Pemy-2006} formulate the liquidation problem as an optimal stopping problem with regime switching, while Pemy et al. \cite{Pemy-2008} study a liquidation problem with a non-oversell constraint, which differs from our terminal state constraint. The constraint in Bian et al \cite{Bian-2016} is an absorption constraint, which does not lead to a backward equation with a singular terminal value. Similarly, the model in Siu et al \cite{Elliott-2019} operates in discrete time and does not yield a backward equation with a singular terminal value either.


The rest of the paper is organized as follows. This section concludes with an introduction to the notation and standing assumptions. In Section \ref{sec:unconstrained}, we consider the unconstrained stochastic control problem where the open terminal position will be penalized. Section \ref{sec:BSDE-existence} establishes the existence result for the BSDE system with singular terminal values. Specifically, we first truncate the terminal value to a finite one and then consider the limit as the truncation approaches infinity. In addition to confirming that this limit is a solution to the BSDE system, we also demonstrate that it is the minimal solution. In Section \ref{sec:wellposedness-control}, we address the stochastic control problem \eqref{cost}-\eqref{state}. More importantly, we establish the uniqueness result for the BSDE system with singular terminal values.

\subsection*{Notation}

{\bf Driving Stochastic Processes.}
Let $(\Omega, \mathcal{F}, \mathbb{P})$ be a complete probability space, and let $W$ denote a standard $n$-dimensional Brownian motion. We denote $N$ as a Poisson random measure $N(\dt,\de)$ on $\mathbb{R}^+ \times \mathcal{E}$, where $\mathcal{E} \subseteq \mathbb{R}^l \setminus \{0\}$ is a Borel set. The Poisson random measure is assumed to be induced by a stationary Poisson point process with a stationary intensity measure $\nu(\de)\dt$, satisfying $\nu(\mathcal{E}) < \infty$. The compensated Poisson random measure is denoted by $\widetilde{N}(\dt,\de)$.

Let $\alpha$ be a continuous-time stationary Markov chain taking values in $\mathcal{M} = \{1, 2, \ldots, \ell\}$ with $\ell \geq 1$. The Markov chain has a generator $Q = (q^{ij})_{\ell \times \ell}$, where $q^{ij} \geq 0$ for $i \neq j$ and $\sum_{j=1}^{\ell} q^{ij} = 0$ for all $i \in \mathcal{M}$. Thus, we have $q^{ii} \leq 0$ for all $i \in \mathcal{M}$.

We assume that $W$, $N$, and $\alpha$ are independent of each other.

{\bf Filtrations}. Denote by $\mathbb F:=(\mathcal F_t)_{t\geq 0}$ the augmented natural filtration of $(W,N,\alpha)$, and by $\mathbb G:=(\mathcal G)_{t\geq 0}$ the augmented natural filtration of $(W,N)$.
Denote by $\mathcal{P}^{\mathbb F}$ (resp. $\mathcal{P}^{\bG}$) the $\sigma$-algebra of $\mathbb{F}$- (resp. $\mathbb{G}$-) predictable subsets of $\Omega\times[0,T]$, and by $\mathcal{B}(\cE)$ the Borel $\sigma$-algebra of $\cE$.

{\bf Spaces.} For each $\hat t\in[0,T]$, each space $\mathcal S$, each filtration $\mathcal H=\mathbb F$ or $\mathbb G$, and each $\sigma$-algebra $\mathcal Q=\mathcal P^{\mathbb F}$ or $\mathcal P^{\mathbb G}$, we define the following spaces of stochastic processes
\begin{align*}
	L^{2}_{\mathcal H}(0, \hat t; \mathcal S )&=\Big\{\phi:[0, \hat t]\times\Omega\rightarrow
	\mathcal S \;\Big|\;\phi\mbox{ is } \mathcal H%
	\mbox{-predictable and }\E\left[\int_{0}^{\hat t}|\phi_t|^{2}\dt\right]<\infty
	\Big\}, \\
	L^{\infty}_{\mathcal H}(0, \hat t;\mathcal S)&=\Big\{\phi:[0, \hat t]\times\Omega
	\rightarrow \mathcal S \;\Big|\;\phi\mbox{ is } \mathcal H%
	\mbox{-predictable and } \esssup_{\omega,t}|\phi_t(\omega)|<\infty \Big\},\\
S^{2}_{\mathcal H}(0, \hat t; \mathcal S )&=\Big\{\phi:[0,\hat t]\times\Omega\to \mathcal S
		\;\Big|\;\phi \mbox{ is c\`adl\`ag, $\mathcal H$-adapted and } \mathbb E\left[\sup_{0\leq t\leq\hat t }|\phi_t|^2\right]<\infty \Big\},\\
		S^{\infty}_{\mathcal H}(0, \hat t; \mathcal S )&=\Big\{\phi:[0,\hat t]\times\Omega\to \mathcal S
	\;\Big|\;\phi \mbox{ is c\`adl\`ag, $\mathcal H$-adapted and } \esssup_{\omega,t}|\phi_t(\omega)|<\infty \Big\},\\
	L^{2}_{\mathcal{Q}}(0, \hat t;\mathcal S)&=\Big\{\phi:[0, \hat t]\times\Omega\times\cE
	\rightarrow \mathcal S \;\Big|\;\phi \mbox{ is }\mathcal{Q}\otimes\mathcal{B}(\cE)%
	\mbox{-measurable and } \E\left[\int_{0}^{\hat t}\int_{\cE}|\phi_t(e)|^{2}\nu(\de)\dt\right]<\infty\Big\},\\
	L^{\infty}_{\mathcal{Q}}(0, \hat t; \mathcal S)&=\Big\{\phi:[0, \hat t]\times\Omega\times\cE
	\rightarrow \mathcal S \;\Big|\;\phi \mbox{ is }\mathcal{Q}\otimes\mathcal{B}(\cE)%
	\mbox{-measurable and } \esssup_{\omega,t}|\phi_t(\omega)|<\infty \Big\}.
\end{align*}
Moreover, we will also use the following space of essentially bounded random variables
\begin{equation*}
	L^{\infty}_{\mathcal{G}_T}=\Big\{\zeta :\Omega\rightarrow
\mathbb{R}\;\Big|\;\zeta \mbox { is }\mathcal{G}_{T}\mbox{-measurable and } \esssup_{\omega}|\zeta(\omega)|<\infty\Big\}.
\end{equation*}
The space of admissible strategies is defined as
\begin{align*}
	\mathcal{A}_0:=\Big \{(\xi,\beta)\in L^2_\mathbb{F}(0,T;\mathbb{R})\times L^{2}_{\mathcal{P}^{\mathbb F}}(0,T;\mathbb{R})\;\Big|\;X^{\xi,\beta}_T=0 \Big\},
\end{align*}
where $X^{\xi,\beta}$ is the position process corresponding to the trading rate $(\xi,\beta)$.

{\bf Convention.} In the above notation of spaces, $(0,\hat t)$ will be omitted if $\hat t=T$. 

The following assumption is in force throughout.
\begin{ass}\label{ass1}
The impact processes satisfy $\eta^i, \lambda^i\in L_{\mathbb{G}}^{\infty}(\R^+)$ and $\gamma^i\in L^{\infty}_{\mathcal{P}^{\mathbb G}}(\R^+)$ for each $i\in\cM$. There exist positive constants $\eta_\star$, $\eta^\star$, $\gamma^\star$ and $\lambda^\star$ such that $\eta_\star\leq \eta^i_t\leq \eta^\star$, $\lambda^i_t\leq \lambda^\star$, $\gamma^i_t(e)\leq \gamma^\star$ for all $(t,e,i)\in[0,T]\times\mathcal E\times \cM$.
\end{ass}

\section{The unconstrained stochastic control problem}\label{sec:unconstrained}
In this section, we consider a stochastic control problem without liquidation constraint, however, the terminal open position is penalized.

For each constant $L>0$, we consider the following unconstrained control problem
\begin{equation}\label{cost-L}
	\begin{split}
		J^L_t(\xi,\beta;x,i)
		:=
		\E\left[ \left. \int_t^T\left(\eta_r^{\alpha_{r-}}|\xi_r|^2+\lambda_r^{\alpha_{r-}}|X_r|^2
		+\int_{\cE}\gamma_r^{\alpha_{r-}}(e)|\beta_r(e)|^2\nu(\de)\right)\dr+L|X_T|^2\;\right|\mathcal F_t\right]\rightarrow\min,
	\end{split}
\end{equation}
subject to
\begin{equation}\label{state-L}
	X_s=x-\int_t^s\xi_r\dr-\int_t^s\int_{\mathcal E}\beta_r(e)\,N(\dr,\de),\quad s\geq t.
\end{equation}
Here, $(x,i)$ in \eqref{cost-L} is the pair of initial state and initial regime starting at time $t$.

Define the value function as
$$V^L_t(x,i):=\inf_{(\xi,\beta)\in L^2_\mathbb{F}(\mathbb{R})\times L^{2}_{\mathcal{P}^{\mathbb F}}(\mathbb{R})}J^L_t(\xi,\beta;x,i)
.$$ When $t=0$, we drop the dependence of $V^L$ and $J^L$ on $t=0$.

By necessary and sufficient stochastic maximum principle, there is a one-to-one correspondence between the optimal control and the value function of control problem \eqref{cost-L}-\eqref{state-L} via
\begin{equation}\label{control-L}
	\hat \xi^{L}_s= \frac{ Y^{L,\alpha_{s-}}_{s} }{ \eta^{\alpha_{s-}}_s } X_{s},\qquad \hat\beta^L_s = \frac{Y^{L,\alpha_{s-}}_{s-}+\Psi^{L,\alpha_{s-}}_s}{\gamma^{\alpha_{s-}}_s+Y^{L,\alpha_{s-}}_{s-}+\Psi^{L,\alpha_{s-}}_{s}} X_{s-}, \qquad s\geq t
\end{equation}
and
\begin{equation}\label{value-L}
	V^L_t( x,i )=Y^{L,i}_tx^2,
\end{equation}
where $(Y^{L,i},Z^{L,i},\Psi^{L,i})$ satisfies the following the BSDE system
\begin{equation}\label{BSDE-L}
	\left\{\begin{split}
		-\dd Y^{i}_t
		= &~ \left( \lambda_t^i-\frac{(Y_t^{i})^2}{\eta_t^i}-\int_{\cE}\frac{(Y^{i}_t+\Psi_t^{i}(e))^2}{\gamma_t^i(e)+Y_t^{i}+\Psi_t^{i}(e)} \mathbf{1}_{ \{ Y^{i}_t+\Psi^{i}_t(e)>0 \} }\,\nu(\de)
		+\sum_{j=1}^{\ell}q^{ij}Y_t^{j}\right) \dt \\
		&~ - (Z_t^{i})^{\top}\dw_t - \int_{\cE}\Psi_t^{i}(e)\widetilde N(\dt,\de), \\
		Y^{i}_T=&~L, \quad i\in\cM.
	\end{split}\right.
\end{equation}
Thus, to prove there exists a solution to the BSDE system \eqref{BSDE-L}, it is sufficient to prove that the control problem \eqref{cost-L}-\eqref{state-L} has a solution.
\begin{lemma}
	The control problem \eqref{cost-L}-\eqref{state-L} has an optimal control $(\xi^*,\beta^*)\in L^2_\mathbb{F}(\mathbb{R})\times L^{2}_{\mathcal{P}^{\mathbb F}}(\mathbb{R}).$
\end{lemma}
\begin{proof}
By \cite[Lemma 5.2]{GHS-2018} (see also \cite[Lemma 3]{KP-2016}), the optimal state process (if exists) has non increasing sample path. Thus, $0\leq X_s\leq x$ and $Y^{L,i}_s\geq 0$, which implies that $Y^{L,i}_{s-}+\Psi^{L,i}_s(e)\geq 0$ by \cite[The proof of Theorem 3.1]{Hu2023}. By \eqref{control-L}, any $\beta$ with $\textrm{Leb}\otimes\nu\left\{(s,e):|\beta_s(e)|>x\right\}>0$ cannot be optimal.

Moreover, the value function $V^L_t(x,i)$ is bounded from both below and above; indeed, $0$ is a lower bound and $(0,0)$ is an admissible strategy leading to a finite cost as an upper bound. Let $(\xi^n,\beta^n)_n$ be a minimizing sequence in $L^2_\mathbb{F}(\mathbb{R})\times L^{2}_{\mathcal{P}^{\mathbb F}}(\mathbb R)$, i.e. $\lim_{n\rightarrow\infty} J^L_t(\xi^n,\beta^n;x,i)=V^L_t(x,i)$. W.l.o.g., we can assume $|\beta^n_s(e)|\leq x$ for each $n$. Then it holds that
\[
	\sup_n\mathbb E\left[ \int_0^T\eta_\star (\xi^n_t)^2\dt+\int_0^T\int_{\mathcal E}(\beta^n_t(e))^2\,\nu(\de)\dt \right]<\infty,
\]
which yields a weakly convergent subsequence in $L^2_\mathbb{F}(\mathbb{R})\times L^{2}_{\mathcal{P}^{\mathbb F}}(\mathbb R)$, that is, $(\xi^{n_k},\beta^{n_k}) \rightharpoonup (\xi^*,\beta^*)$ in $L^2_\mathbb{F}(\mathbb{R})\times L^{2}_{\mathcal{P}^{\mathbb F}}(\mathbb R)$.

It can be verified that the mapping $(\xi,\beta)\mapsto J^L_t(\xi,\beta;x,i)$ is convex and strongly continuous in $L^2_{\mathbb F}(\mathbb R)\times L^{2}_{\mathcal{P}^{\mathbb F}}(\mathbb{R})$. Thus, it is weakly lower semicontinuous. It implies that
\[
	V^L_t(x,i)\leq J^L_t(\xi^*,\beta^*;x,i)\leq	\liminf_{k\rightarrow\infty}J^L_t(\xi^{n_k},\beta^{n_k};x,i) = V^L_t(x,i),
\]
and thus, $(\xi^*,\beta^*)$ is an optimal control.
\end{proof}

\begin{corollary}\label{coro:BSDE-L}
The system of BSDEs \eqref{BSDE-L} has a solution $(Y^{L,i},Z^{L,i},\Psi^{L,i})_{i\in\mathcal M}\in S^2_{\mathbb G}(\mathbb R^\ell)\times L^2_{\mathbb G}(\R^{n\ell})\times L^{2}_{\mathcal P^{\mathbb G}}(\R^\ell)$. Moreover, the $Y$-component is nonnegative.
\end{corollary}

\begin{remark}
	Note that the recent work \cite{SX2024} studies the system \eqref{BSDE-L} in detail as a special case (see \cite[Theorem 3.1 and Theorem 3.2]{SX2024}), while ours is a pure existence result that is sufficient for our purpose.
\end{remark}

\section{The existence result of the BSDE system}\label{sec:BSDE-existence}
From Section \ref{sec:unconstrained}, the system of BSDEs we are interested in is
\begin{equation}\label{BSDE}
	\left\{\begin{split}
	-\dd Y_t^i=&~ \left( \lambda_t^i-\frac{(Y_t^i)^2}{\eta_t^i}-\int_{\cE}\frac{(Y^i_{t}+\Psi_t^i(e))^2}{\gamma_t^i(e)+Y_{t}^i+\Psi_t^i(e)} \mathbf{1}_{ \{Y^i_{t}+\Psi^i_t(e)>0\} } 	\,\nu(\de)
	+\sum_{j=1}^{\ell}q^{ij}Y_t^j\right)\dt\\
	&~ - (Z_t^i)^{\top}\dw_t
	 - \int_{\cE}\Psi_t^i(e)\widetilde N(\dt,\de),\\
	 	\lim_{t\nearrow T}Y_t^i=&~+\infty, \quad i\in\cM,
	\end{split}\right.
\end{equation}
where the singular terminal condition is because of the liquidation constraint $X_T=0$. Due to the singularity, the solution is understood in the following sense.
\begin{definition}\label{def:solution}
	A system of vector processes $\left(Y^i,Z^i,\Psi^i\right)_{i\in\cM}$ is called a solution to the system of BSDEs \eqref{BSDE}, if it satisfies
	\begin{itemize}
		\item for all $0\leq s\leq t<T$ and $i\in\cM$:
		\begin{equation}
			\label{Riccatilocal}
			\begin{split}
			Y_s^i=&~Y_t^i+\int_s^t f^i(r,Y^i_r, Y^{-i}_r,\Psi^i_r)\dr 
			 -\int_s^t (Z_r^i)^{\top}\dw_r-\int_{s}^t\int_{\cE}\Psi_r^i(e)\widetilde N(\dr,\de),
			\end{split}
	\end{equation}
where the driver is defined as
\begin{align}\label{def:f}
	f^i(t,y^i,y^{-i},\psi)=\lambda^i_t-\frac{(y^i)^2}{\eta^i_t}-\int_{\cE}\frac{(y^i+\psi(e))^2}{\gamma^i_t(e)+y^i+\psi(e)} 		\mathbf{1}_{ \{y^i+\psi(e)>0\}		 }						 \,\nu(\de)
	+\sum_{j=1}^{\ell}q^{ij}y^j,
\end{align}
	 for any vector $y\in\R^\ell$, $y^i$ stands for the $i$-th component of $y$, and $y^{-i}=(y^{1}, \ldots, y^{i-1}, y^{i+1}, \ldots, y^{\ell})$.
		\item for each $0\leq t<T$ and $i\in\cM$, the triplet $(Y^i,Z^i,\Psi^i)\in S^{2}_{\mathbb G}(0,t;\R)\times L^2_{\mathbb G}(0,t;\R^n)\times L^{2}_{\mathcal P^{\mathbb G}}(0,t;\R)$;
		\item for each $i\in\cM, \ \lim_{t\nearrow T}Y_t^i=+\infty$ a.s..
	\end{itemize}
\end{definition}

\begin{theorem}\label{thm:existence}
	Under Assumption \ref{ass1},
	the system of BSDEs \eqref{BSDE} admits one solution $(Y^i,Z^i,\Psi^i)_{i\in\mathcal M}$ in the sense of Definition \ref{def:solution}. Moreover, the $Y$-component of this solution is strictly positive and $\Psi^i\in L^\infty_{\mathcal P^{\mathbb G}}(0,t;\mathbb R)$ for each $t<T$ and $i\in\mathcal M$.
\end{theorem}
To prove Theorem \ref{thm:existence}, we truncate the singular terminal condition by $L>0$ and the resulting system of BSDEs with standard terminal conditions has a solution by Corollary \ref{coro:BSDE-L}. We then consider the limit as the penalization degree $L$ approaches infinity. To ensure the convergence is rigorous, we require careful estimates to establish both an upper and a lower bound for the truncated system of BSDEs \eqref{BSDE-L}.
To achieve this, we consider the following comparison principle for multi-dimensional BSDEs with monotone drivers. The proof is inspired by \cite[Theorem 2.1]{Hu2023} and is provided in the appendix for completeness. We emphasize that the following comparison theorem is essential not only for the existence result but also for the subsequent uniqueness result.
\begin{proposition}
	\label{prop:comparison}
	For each $i\in\mathcal M$, assume $(Y^i, Z^i,\Psi^i)$ and $(\overline Y^i, \overline Z^i,\overline\Psi^i)\in S^{2}_{\mathbb G}(\R)\times L^2_{\mathbb G}(\R^n)\times L^{2}_{\mathcal P^{\mathbb G}}(\R)$, and assume $(Y^i, Z^i,\Psi^i)_{i\in\mathcal M}$ and $(\overline Y^i, \overline Z^i,\overline\Psi^i)_{i\in\mathcal M}$
	 satisfy the following two systems of BSDEs
	\begin{align*}
		Y^{i}_{t}=\zeta^{i}+\int_t^T f^i(s, Y^{i}_{s}, Y^{-i}_{s}, \Psi^i_s)\ds-\int_t^T (Z^{i}_{s})^{\top}\dw_s-\int_{t}^T\Psi^i_s(e)\,\widetilde N(\ds,\de), \quad \mbox{ $i\in\cM$},
	\end{align*}
	respectively,
	\begin{align*}
		\overline Y^{i}_{t}=\overline\zeta^{i}+\int_t^T \overline f^i(s, \overline Y^{i}_{s}, \overline Y^{-i}_{s}, \overline\Psi^i)\ds-\int_t^T (\overline Z^{i}_{s})^{\top}\dw_{s}
		-\int_{t}^T\overline\Psi^i_s(e)\,\widetilde N(\ds,\de), \quad \mbox{ $i\in\cM$.}
	\end{align*}
	Moreover, assume that there exists a constant $c>0$ such that, for all $i\in\cM$,
	\begin{enumerate}
		\item $\zeta^{i}, \ \overline\zeta^{i}\in L^\infty_{\mathcal G_T}$, and $\zeta^{i}\leq\overline\zeta^{i}$;
		\item
		$f^i(s, y^i, y^{-i}, \psi )-f^i(s, y^i, \overline y^{-i}, \psi ) \leq c \sum_{j\neq i} (y^{j}-\overline y^{j} )^+ $,
		for any $0\leq s<T$, $y^i\in\R$, $y^{-i}, \overline y^{-i}\in\mathbb{R}^{\ell-1}$ and $\psi$ that is $\nu$ square integrable;
		\item For any $0\leq s<T$,
		$
			 ( Y^i_s-\overline Y^i_s)^+\left(f^i(s, Y^i_s,Y^{-i}_s, \overline\Psi^i_s)-f^i(s, \overline Y^i_s, Y^{-i}_s, \overline\Psi^i_s)\right) \leq c \left[(Y^i_s-\overline Y^i_s)^+\right]^2;
		$
		\item
	
	 For any $0\leq s<T$, $y\in\mathbb R^\ell$ and $\psi,\psi'$ that are $\nu$ square integrable,	$ f^i\left(s,y, \psi\right)-f^i\left(s,y,\psi'\right)
			\leq \int_{\cE} \left|\psi(e)-\psi'(e)\right| \,\nu(\de)
	$;
		\item For any $0\leq s<T$, $ f^i\left(s, \overline Y^{i}_{s}, \overline Y^{-i}_{s} ,\overline\Psi^i_s\right)\leq \overline f^i\left(s, \overline Y^{i}_{s}, \overline Y^{-i}_{s}, \overline\Psi^i_s\right)$.
	\end{enumerate}
	Then $Y^{i}_{t}\leq \overline Y^{i}_{t}$, for a.e. $t\in[0, T]$ and all $i\in\cM$.
\end{proposition}

\begin{remark}
	(1) Note that Conditions 2 and 4 are applied to any argument while
	Conditions 3 and 5 are only applied to solutions of the BSDE systems.
	
	(2) In Lemma \ref{lem:bound-Y}, the above comparison principle will be applied to the truncated system of BSDEs \eqref{BSDE-L} with $L>0$
	to get upper and lower bounds for the convergence purpose.
	Moreover, it holds, for any $t\in[0,T]$, any $y\in\mathbb R$ and any $\psi,\psi'$ that are $\nu$ square integrable, that
	\begin{multline}\label{Condition-4}
		-\int_{\mathcal E} \frac{(y+\psi(e))^2}{\gamma_t(e)+y+\psi(e)} \mathbf{1}_{ \{y+\psi(e)>0\} }\,\nu(\de) + \int_{\mathcal E} \frac{(y+\psi'(e))^2}{\gamma_t(e)+y+\psi'(e)} \mathbf{1}_{ \{y+\psi'(e)>0\} }\,\nu(\de)\\
		 \leq \int_{\mathcal E} | \psi(e) - \psi'(e) |\,\nu(\de).
	\end{multline}
	Thus, Condition 4 in Proposition \ref{prop:comparison} always holds.
\end{remark}

Recall that $(Y^{L,i},Z^{L,i},\Psi^{L,i})_{i\in\mathcal M}$ is a solution to the BSDE system \eqref{BSDE-L}. The next lemma establishes the upper and lower bounds for $Y^{L,i}$.
\begin{lemma}\label{lem:bound-Y}
	Under Assumption \ref{ass1}, for each $t\in[0,T)$ and $i\in\cM$, we have the following estimate for $Y^{L,i}_t$:
	\begin{align}\label{lowerupper}
		\frac{1}{(1+L^{-1})e^{\check c(T-t)}-1}\leq Y^{L,i}_t&\leq \frac{\eta^\star}{T-t}+\frac{\lambda^\star}{3}(T-t),
	\end{align}
where
$\check c$ is any positive constant satisfying $\max\{\eta_\star^{-1}, \nu(\cE)\}\leq \check c$ and the positive contants $\eta_\star$, $\eta^\star$ and $\lambda^\star$ are given in Assumption \ref{ass1}.

Moreover, it holds that for each $i\in\mathcal M$
\begin{equation}\label{estimate:Y+Psi}
		Y^{L,i}_{t-}+\Psi^{L,i}_t(e)\geq \frac{1}{(1+L^{-1})e^{\check c(T-t)}-1}, \quad \dd\mathbb{P}\otimes\dt\otimes\dd\nu-a.e.
\end{equation}
\end{lemma}
\begin{proof}
{\bf The upper bound for $Y^{L,i}$.}
From Young's inequality and Assumption \ref{ass1}, for any $t\in[0,T)$ and $i\in\mathcal M$ we have
\begin{align}\label{upper}
	\lambda_t^i-\frac{y^2}{\eta_t^i}\leq \lambda^\star-2\frac{y}{T-t}+\frac{\eta_t^i}{(T-t)^2}\leq \lambda^\star-2\frac{y}{T-t}+\frac{\eta^\star}{(T-t)^2}.
\end{align}
Let $\epsilon\in(0,T)$, and $\widetilde c>0$ be such that $\max_{i\in\cM}\esssup_{t,\omega} Y^{L,i}_{t}\leq \tilde c$, and consider the following system of linear ODEs on $[0,T-\epsilon]$:
\begin{equation}\label{Riccatiupper}
	\left\{\begin{split}
		\dd \overline Y_t^{L,i} =&~-\left(\lambda^\star-2\frac{\overline Y^{L,i}_t}{T-t}+\frac{\eta^\star}{(T-t)^2}+\sum_{j=1}^{\ell}q^{ij}\overline Y^{L,j}_t\right)\dt,\\
		\overline Y^{L,i}_{T-\epsilon}=&~\widetilde c, \quad i\in\cM,
	\end{split}\right.
\end{equation}
which admits a unique solution $(\overline Y^{L,i})_{i\in\cM}$:
\begin{align*}
\overline Y^{L,i}_t=\frac{1}{\Gamma_t}\left(\Gamma_{T-\epsilon} \widetilde c
+ \int_t^{T-\epsilon}\Gamma_s\Big( \lambda^\star+\frac{\eta^\star}{(T-s)^2}\Big)\ds\right), \quad i\in\cM,
\end{align*}
where
$
\Gamma_t
=\Big(\frac{T-t}{T}\Big)^2.
$
Note that $\overline Y^{L,i}$ is independent of $i$.

We will verify that the BSDE systems \eqref{BSDE-L} and \eqref{Riccatiupper} satisfy the conditions of Proposition \ref{prop:comparison}.

Condition $1$ is clear. Condition $2$ holds because $q^{ij}\geq 0$ for $j\neq i$, and Condition $5$ is given by \eqref{upper}. For Condition $3$, we only need to consider the case $Y^i_t>\overline Y^i_t$, under which we have
\begin{equation*}
	\begin{split}
		&~(Y^i_t - \overline Y^i_t	)\left( -\frac{		(Y^i_t)^2	}{\eta^i_t}+\frac{ (\overline Y^i_t)^2				}{\eta^i_t}		- \int_{\mathcal E} \frac{ (Y^i_t)^2 }{\gamma^i_t(e)+Y^i_t}\nu(\de)	+ \int_{\mathcal E} \frac{ (\overline Y^i_t)^2 }{ \gamma^i_t(e) + \overline Y^i_t } \nu(\de)			+q^{ii}( Y^i_t-\overline Y^i_t )		 \right) \\
		\leq&~ -\left(\frac{ Y^i_t+\overline Y^i_t }{\eta^i_t} -q^{ii}\right) (Y^i_t-\overline Y^i_t)^2- \int_{\mathcal E} (Y^i_t - \overline Y^i_t	)\left( \frac{ (Y^i_t)^2 }{\gamma^i_t(e)+Y^i_t} 	- \frac{ (\overline Y^i_t)^2 }{ \gamma^i_t(e) + \overline Y^i_t } \right) \nu(\de)				 \\
		\leq&~0,
	\end{split}
\end{equation*}
where the last inequality is obtained by noticing that $\frac{Y^i_t+\overline Y^i_t}{\eta^i_t}-q^{ii}\geq 0$ and that the mapping $y\mapsto \frac{y^2}{\gamma^i_t(e)+y}$ is increasing. Moreover, Condition 4 holds due to \eqref{Condition-4}. 

Thus, by applying Proposition \ref{prop:comparison} to the BSDE systems \eqref{BSDE-L} and \eqref{Riccatiupper}, we get, for each $t\in[0,T-\epsilon]$,
\begin{align*}
Y^{L,i}_t\leq \overline Y^{L,i}_t=\frac{1}{(T-t)^2}\Big[\epsilon^2 \widetilde c + \int_t^{T-\epsilon}\left(\eta^\star+\lambda^\star(T-s)^2\right)\ds\Big].
\end{align*}
Sending $\epsilon\downarrow 0$, we get, for each $t\in[0,T)$,
\begin{align*}
Y^{L,i}_t&\leq \frac{ 1}{(T-t)^2}\int_t^{T}\left(\eta^\star+\lambda^\star(T-s)^2\right)\ds=\frac{\eta^\star}{T-t}+\frac{\lambda^\star}{3}(T-t).
\end{align*}
{\bf The lower bound for $Y^{L,i}$.} Let $\check c$ be a positive constant such that $\max\{\eta_\star^{-1}, \nu(\cE)\}\leq \check c$.

Noting that $\lambda^i\geq 0$, and that
\begin{equation}\label{solution:lower}
	\underline Y^{L,i}_t=\frac{1}{(1+L^{-1})e^{\check c(T-t)}-1} > 0
\end{equation}
is the unique solution to the following system of linear ODEs
\begin{equation}\label{Riccatilower}
	\left\{\begin{split}
		\dd \underline Y^{L,i}_t=&~-\bigg(-\check c(\underline Y^{L,i}_t)^2-\check c\underline Y^{L,i}_t+\sum_{j=1}^{\ell}q^{ij}\underline Y^{L,j}_t\bigg)\dt,\\
		\underline Y^{L,i}_T=&~L, 		~ i\in\cM.
	\end{split}\right.
\end{equation}
To apply Proposition \ref{prop:comparison}, we will compare $\widetilde Y^i:=-Y^{L,i}$ and $\widetilde {\underline Y}^i:=-\underline Y^{L,i}$, which satisfy the BSDE and ODE with drivers $h^i(t,y,z,\psi)=-\lambda^i_t+\frac{(y^i)^2}{\eta^i_t}+\int_{\mathcal E} \frac{( y^i+\psi(e)	 )^2}{\gamma^i_t(e)-y^i-\psi(e)} \mathbf{1}_{ \{ y^i+\psi(e)<0 \} } \nu(\de) + \sum_{j=1}^\ell q^{ij} y^j$, respectively, $\overline h^i(t,y,z,\psi)=\check c (y^i)^2-\check c y^i + \sum_{j=1}^\ell q^{ij} y^j$.
It can be verified directly that Conditions $1$-$2$ in Proposition \ref{prop:comparison} are satisfied. To verify Condition $3$, it is sufficient to consider the case $\widetilde Y^i>\widetilde {\underline Y}^i$:
\begin{equation*}
	\begin{split}
		&~(\widetilde Y^i_t-\widetilde {\underline Y}^i_t ) \left(		\frac{ (\widetilde Y^i_t)^2 - (\widetilde{\underline Y}^i_t)^2 }{\eta^i_t}	+\int_{\mathcal E} 		\frac{ ( \widetilde Y^i_t )^2 }{\gamma^i_t(e)-\widetilde Y^i_t }	\,\nu(\de)-\int_{\mathcal E} 	\frac{(	\widetilde{\underline Y}^i_t )^2}{\gamma^i_t(e) - \widetilde{\underline Y}^i_t }			 \,\nu(\de)			+ q^{ii}( \widetilde Y^i_t-\widetilde{\underline Y}^i_t ) 	\right) \\
		\leq&~ \frac{\widetilde Y^i_t+\widetilde{\underline Y}^i_t}{\eta^i_t}(\widetilde Y^i_t-\widetilde{\underline Y}^i_t)^2 + q^{ii}(	 \widetilde Y^i_t-\widetilde {\underline Y}^i_t 	)^2 +(\widetilde Y^i_t-\widetilde {\underline Y}^i_t ) \int_{\mathcal E} \left(	\frac{ ( \widetilde Y^i_t )^2 }{\gamma^i_t(e)-\widetilde Y^i_t } - \frac{(	\widetilde{\underline Y}^i_t )^2}{\gamma^i_t(e) - \widetilde{\underline Y}^i_t }		\right)			 \nu(\de)\\
		\leq&~ -\frac{1}{\eta^\star}\left(\delta+\frac{1}{ (1+L^{-1})e^{\check cT} -1 }\right)(\widetilde Y^i_t-\widetilde{\underline Y}^i_t)^2,
	\end{split}
\end{equation*}
where the last inequality is due to $q^{ii}\leq0$, the fact that the mapping $x\mapsto \frac{x^2}{\gamma^i_t(e)-x}$ is decreasing on $(-\infty,0)$ for each given $(t,e)$, and the fact $0>\widetilde Y^i_t >\widetilde{\underline Y}^i_t $. 
Condition 4 holds by \eqref{Condition-4} and Condition $5$ holds by using the assumption for $\check c$.

Thus, Proposition \ref{prop:comparison} implies that 
\begin{align*}
	Y^{L,i}_t\geq \underline Y^{L,i}_t=\frac{1}{(1+L^{-1})e^{\check c(T-t)}-1}.
\end{align*}
The estimate \eqref{estimate:Y+Psi} can be obtained by a similar argument as in the proof of \cite[Theorem 3.1]{Hu2023}.
\end{proof}
Now we are ready to prove Theorem \ref{thm:existence}.
\begin{proof}[Proof of Theorem \ref{thm:existence}.]
	{\bf The convergence of $Y^{L,i}$.} By Proposition \ref{prop:comparison} we know that $Y^{L,i}\leq Y^{N,i}$ for $L\leq N$. By \eqref{lowerupper}, for fixed $t<T$, the family of random variables $\{Y^{L,i}_t\}_{L\geq0,i\in\cM}$ is bounded from above.
Hence, for any $t<T$ and $i\in\cM$, we can define $Y^{i}_t$ as the increasing limit of $Y^{L,i}_t$ as $L\rightarrow\infty$. The limit $Y^i_t$ inherits the upper bound of $Y^{L,i}_t$ in \eqref{lowerupper}, i.e.,
$
	Y^{i}_t\leq \frac{\eta^\star}{T-t}+\frac{\lambda^\star}{3}(T-t).$
By dominated convergence, we have
\begin{equation}\label{convergence-Y-1}
	\lim_{ L\rightarrow\infty} \left( \mathbb E\left[| Y^{L,i}_t - Y^i_t |^p\right]+\mathbb E\left[\int_0^t|Y^{L,i}_r-Y^i_r|^p\dr\right] \right) = 0,
\end{equation}
for each $i\in\mathcal M$, $t<T$ and $p\geq 1$.

Sending $L$ to infinity in the lower bound in \eqref{lowerupper} yields
$
	Y^{i}_t\geq\frac{1}{e^{\check c(T-t)}-1},
$
which implies that $Y^i$ satisfies the singular terminal condition in \eqref{BSDE}.

{\bf The convergence of $(Z^{L,i},\Psi^{L,i})$.} Let $0\leq s\leq t<T$. For $N,L\geq0$, we put
\begin{align*}
	\Delta Y^i=Y^{N,i}-Y^{L,i}, \ \Delta Z^i=Z^{N,i}-Z^{L,i}, \ \Delta \Psi^i=\Psi^{N,i}-\Psi^{L,i}.
\end{align*}
Applying It\^o's formula to $|\Delta Y^i|^2$ yields,
\begin{equation}
	\label{itosquare}
	\begin{split}
	&~ |\Delta Y^i_s|^2+\int_s^t|\Delta Z^i_r|^2\dr+\int_s^t\int_{\cE}|\Delta\Psi^i_r(e)|^2N(\dr,\de) \\
	=&~ |\Delta Y^i_t|^2+2\int_s^t\Delta Y^i_r\Big(f^i(r,Y^{N,i}_r,Y^{N,-i}_r,\Psi^{N,i}_r)-f^i(r,Y^{L,i}_r,Y^{L,-i}_r,\Psi^{L,i}_r)\Big)\dr \\
	&~-2\int_s^t\Delta Y^i_r(\Delta Z^i_r)^{\top}\dw_r-2\int_s^t\int_{\cE}\Delta Y^i_r\Delta\Psi^i_r(e)\widetilde N(\dr,\de),
	\end{split}
\end{equation}
where we recall the driver $f^i$ is defined in \eqref{def:f}.
By noting that if $u>0$ or $u<-2<-2\frac{Y^{L,i}_r+\Psi^{L,i}_r(e)}{\gamma^i_r(e)+Y^{L,i}_r+\Psi^{L,i}_r(e)}$, we have
$$
	\left(\gamma^i_r(e)+Y^{L,i}_r+\Psi^{L,i}_r(e)\right)u^2+2\left(Y^{L,i}_r+\Psi^{L,i}_r(e)\right)u> 0,
$$
so
\begin{equation*}
	\begin{split}
	-\frac{\left(Y^{L,i}_r+\Psi^{L,i}_r(e)\right)^2}{\gamma^i_r(e)+Y^{L,i}_r+\Psi^{L,i}_r(e)}
	=&~\inf_{u\in\R}\Big[\left(\gamma^i_r(e)+Y^{L,i}_r+\Psi^{L,i}_r(e)\right)u^2+2\left(Y^{L,i}_r+\Psi^{L,i}_r(e)\right)u\Big]\\
	=&~\inf_{-2\leq u\leq 0}\Big[\left(\gamma^i_r(e)+Y^{L,i}_r+\Psi^{L,i}_r(e)\right)u^2+2\left(Y^{L,i}_r+\Psi^{L,i}_r(e)\right)u\Big].
	\end{split}
\end{equation*}
The same equalities hold if $(Y^{L,i},\Psi^{L,i})$ is replaced by $(Y^{N,i},\Psi^{N,i})$. Thus, we have
\begin{equation}
	\label{estimate1}
	\begin{split}
	&~\left|-\frac{\left(Y^{N,i}_r+\Psi^{N,i}_r(e)\right)^2}{\gamma^i_r(e)+Y^{N,i}_r+\Psi^{N,i}_r(e)}
	+\frac{\left(Y^{L,i}_r+\Psi^{L,i}_r(e)\right)^2}{\gamma^i_r(e)+Y^{L,i}_r+\Psi^{L,i}_r(e)}\right| \\
	=&~\left|\inf_{-2\leq u\leq 0}\Big[\left(\gamma^i_r(e)+Y^{N,i}_r+\Psi^{N,i}_r(e)\right)u^2+2\left(Y^{N,i}_r+\Psi^{N,i}_r(e)\right)u\Big] \right.\\
	&~~ \left.-\inf_{-2\leq u\leq 0}\Big[\left(\gamma^i_r(e)+Y^{L,i}_r+\Psi^{L,i}_r(e)\right) u^2+2\left(Y^{L,i}_r+\Psi^{L,i}_r(e)\right)u\Big]\right| \\
	\leq&~ \left(|\Delta Y^i_r|+|\Delta\Psi^i_r(e)|\right)\sup_{-2\leq u\leq 0}(u^2 +2|u|) \\
	=&~8\left(|\Delta Y^i_r|+|\Delta\Psi^i_r(e)|\right).
	\end{split}
\end{equation}
Taking \eqref{estimate1} into \eqref{itosquare}, and taking expectations on both sides of \eqref{itosquare}, 
for a constant $\widehat c>8+8\nu(\mathcal E)+\max_{i\neq j}q^{ij}$, we have for any $0\leq s<t<T$
\begin{equation*}
	\begin{split}
	&~\E\left[\int_s^t|\Delta Z^i_r|^2\dr +\int_s^t\int_{\cE}|\Delta\Psi^i_r(e)|^2\,\nu(\de)\dr\right]\\
	\leq&~ \E\left[|\Delta Y^i_t|^2\right]+ 2\widehat c\sum_{j= 1}^{\ell}\E\left[\int_s^t|\Delta Y^i_r||\Delta Y^j_r|\dr\right] +2\widehat c\E\left[\int_{s}^{t}\int_{\cE}|\Delta Y^i_r||\Delta \Psi^i_r(e)|\,\nu(\de)\dr\right]\\
	\leq&~ \E\left|\Delta Y^i_t|^2\right]+\widehat c \sum_{j= 1}^{\ell}\E\left[\int_s^t|\Delta Y^j_r|^2\dr\right]+\left(\widehat c\ell+2\widehat c^2\nu(\cE)\right)\E\left[\int_s^t|\Delta Y^i_r|^2\dr\right] \\
	&+ \frac{1}{2\nu(\cE)}\E\left[ \int_{s}^{t}\Big(\int_{\cE}|\Delta \Psi^i_r(e)|\,\nu(\de)\Big)^2\dr\right]\\
	\leq&~ \E\left|[ \Delta Y^i_t|^2\right]+\left(\widehat c+\widehat c\ell+2\widehat c^2\nu(\cE)\right)\sum_{j= 1}^{\ell}\E\left[\int_s^t|\Delta Y^j_r|^2\dr\right]+\frac{1}{2}\E\left[\int_{s}^{t}\int_{\cE}|\Delta \Psi^i_r(e)|^2\,\nu(\de)\dr\right],
	\end{split}
\end{equation*}
where we used the Cauchy-Schwarz inequality to get last term.
It implies that
\begin{equation}\label{estimate2}
	\begin{split}
	&~\E\left[\int_s^t|\Delta Z^i_r|^2\dr +\frac{1}{2}\int_s^t\int_{\cE}|\Delta\Psi^i_r(e)|^2\,\nu(\de)\dr\right]\\
	\leq&~ \E\left[|\Delta Y^i_t|^2\right]+ \left(\widehat c+\widehat c\ell+2\widehat c^2\nu(\cE)\right) \sum_{j= 1}^{\ell}\E\left[\int_s^t|\Delta Y^j_r|^2\dr\right].
	\end{split}
\end{equation}
By \eqref{convergence-Y-1}, the right hand side of \eqref{estimate2} converges to zero as $N,L\rightarrow\infty$.
It implies that $(Z^{L,i},\Psi^{L,i})_L$ is a Cauchy sequence in $L^{2}_{\mathbb{G}}(0,t;\R^{n})\times L^{2}_{\mathcal{P}^{ \mathbb G }}(0, t;\R)$ and converges to some $(Z^{i,(t)},\Psi^{i,(t)})\in L^{2}_{\mathbb{G}}(0,t;\R^{n})\times L^{2}_{\mathcal{P}^{\mathbb G}}(0, t;\R)$ for any $t<T$. By uniqueness of the limit, we have $(Z^{i,(t_1)},\Psi^{i,(t_1)})=(Z^{i,(t_2)},\Psi^{i,(t_2)})$ for any $0\leq t_1<t_2<T$. Thus, we get a compatible limit denoted by $(Z^{i},\Psi^{i})$.

{\bf Completing the proof.} By the strong convergence of $(Z^{L,i},\Psi^{L,i})_L$, standard estimate implies that for any $0\leq t<T$,
\begin{equation}\label{sup-convergence-Y}
\lim_{L\rightarrow\infty}\E\left[\sup_{s\in[0,t]}\left|Y^{L,i}_s-Y^{i}_s\right|\right]=0, ~~ i\in\cM.
\end{equation}

Finally, letting $L\rightarrow\infty$ in \eqref{BSDE-L} implies that $(Y^i,Z^i,\Psi^i)$ satisfies \eqref{Riccatilocal} for any $0\leq s\leq t<T$. Since $Y^i$ is bounded on $[0,t]$ for any $t<T$, we know that $\Psi^{i}\in L^{\infty}_{\mathcal{P}^{\mathbb G}}(0, t;\R)$; see e.g. \cite[Corollary 1]{Morlais2009}.
\end{proof}
%
%
The next proposition shows that the solution obtained in Theorem \ref{thm:existence} is minimal. In Section \ref{sec:wellposedness-control}, we will prove that this minimal solution is indeed the unique one; refer to Theorem \ref{thm:uniqueness}, where we will use the following minimality result.
\begin{proposition}\label{prop:minimal}
	The solution $(Y^i,Z^i,\Psi^i)_{i\in\mathcal M}$ obtained in Theorem \ref{thm:existence} is minimal in the following sense: if $(\overline Y^i,\overline Z^i,\overline \Psi^i)_{i\in\cM}$ is another solution of \eqref{BSDE} such that $\overline Y^i$ is bounded from below by a stochastic process $S\in S^2_{\mathbb G}(\mathbb R)$, then $Y^i_t\leq \overline Y^i_t$ a.s. for any $t\in[0,T)$ and $i\in\mathcal M$.
\end{proposition}
\begin{proof}
Fix $L>0$ and let $(Y^{L,i},Z^{L,i},\Psi^{L,i})_{i\in\cM}$ denote the solution of \eqref{BSDE-L}. Let $(\overline Y^i,\overline Z^i,\overline \Psi^i)_{i\in\cM}$ be any solution of \eqref{BSDE} with a lower bound $S\in S^2_{\mathbb G}(\mathbb R)$. Set
\begin{align*}
	\Delta Y^i=Y^{L,i}-\overline Y^{i}, \quad \Delta Z^i=Z^{L,i}-\overline Z^{i}, \quad \Delta \Psi^i=\Psi^{L,i}-\overline \Psi^{i}.
\end{align*}
Let $a>0$ be a constant to be determined later.
Applying the Meyer-It\^o formula \cite[Theorem 70]{Protter2005} to $e^{a\cdot}|(\Delta Y^i_\cdot)^+|^2$, we have for any $0\leq s\leq t<T$,
\begin{equation*}
	\begin{split}
	&~ e^{as}|(\Delta Y^i_s)^+|^2+a\int_s^t e^{ar}|(\Delta Y^i_r)^+|^2\dr+\int_s^t\mathbf{1}_{\{\Delta Y^i_{r-}>0\}}e^{ar}|\Delta Z^i_r|^2\dr \\
	&~+\int_s^t\int_{\cE}e^{ar}\left[\left((\Delta Y^i_{r-}+\Delta\Psi^i_r(e))^+\right)^2-\left((\Delta Y^i_{r-})^+\right)^2\right]N(\dr,\de) \\
	=&~ e^{at}|(\Delta Y^i_t)^+|^2+2\int_s^t e^{ar}(\Delta Y^i_r)^+\left(f^i(r,Y^{L,i}_r,Y^{L,-i}_r,\Psi^{L,i}_r)-f^i(r,\overline Y^{i}_r,\overline Y^{-i}_r,\overline \Psi^{i}_r)\right)\dr \\
	&~-2\int_s^te^{ar}\Delta Y^i_r(\Delta Z^i_r)^{\top}\dw_r+2\int_s^t\int_{\cE}e^{ar}(\Delta Y^i_{r})^+\Delta\Psi^i_r(e)\,\nu(\de)\dr.
	\end{split}
\end{equation*}
Taking conditional expectations, we have
\begin{equation}\label{estimate4}
	\begin{split}
	&~ e^{as}\left|(\Delta Y^i_s)^+\right|^2+a\E\left[ \left. \int_s^t e^{ar}\left|(\Delta Y^i_r)^+\right|^2\dr \;\right|\mathcal G_s \right] \\
	\leq&~ \E\left[\left.e^{at}|(\Delta Y^i_t)^+|^2+2\int_s^t e^{ar}(\Delta Y^i_r)^+\left(f^i(r,Y^{L,i}_r,Y^{L,-i}_r,\Psi^{L,i}_r)-f^i(r,\overline Y^{i}_r,\overline Y^{-i}_r,\overline \Psi^{i}_r)\right)\dr\;\right|\mathcal G_s\right] \\
	&~ -\E\left[ \left. \int_s^t\int_{\cE}e^{ar}\Big(\left((\Delta Y^i_{r}+\Delta\Psi^i_r(e))^+\right)^2-\left((\Delta Y^i_{r})^+\right)^2-2\left(\Delta Y^i_{r}\right)^+\Delta\Psi^i_r(e)\Big)\,\nu(\de)\dr \;\right|\mathcal G_s \right].
	\end{split}
\end{equation}
Next we will estimate the second term on the right hand side of \eqref{estimate4}.
First, similarly to \eqref{estimate1}, we have
\begin{equation*}\label{estimate3}
	\begin{split}
	& -\frac{(Y^{L,i}_r+\Psi^{L,i}_r(e))^2}{\gamma^i_r(e)+Y^{L,i}_r+\Psi^{L,i}_r(e)}
	+\frac{(\overline Y^{i}_r+\overline \Psi^{i}_r(e))^2}{\gamma^i_r(e)+\overline Y^{i}_r+\overline \Psi^{i}_r(e)} \\
	=&~\inf_{-2\leq u\leq 0}\Big((\gamma^i_r(e)+Y^{L,i}_r+\Psi^{L,i}_r(e))u^2+2(Y^{L,i}_r+\Psi^{L,i}_r(e))u\Big)\\
	&~-\inf_{-2\leq u\leq 0}\Big((\gamma^i_r(e)+\overline Y^{i}_r+\overline \Psi^{i}_r(e))u^2+2(\overline Y^{i}_r+\overline \Psi^{i}_r(e))u\Big) \\
	\leq&~ \sup_{-2\leq u\leq 0}\Big((\Delta Y^i_r+\Delta\Psi^i_r(e))(u^2 +2u)\Big) \\
	=&~-\left(\Delta Y^i_r+\Delta\Psi^i_r(e)\right)\mathbf{1}_{\{\Delta Y^i_r+\Delta\Psi^i_r(e)\leq0\}}.
	\end{split}
\end{equation*}
Second, it holds that
\begin{equation*}
	\begin{split}
	(\Delta Y^i_r)^+\left(-\frac{(Y^{L,i}_r)^2}{\eta^i_r}+\frac{(\overline Y^i_r)^2}{\eta^i_r}+\sum_{j=1}^{\ell}q^{ij}\Delta Y^i_r\right) \leq \left(\Delta Y^i_r\right)^+\sum_{j\neq i}q^{ij}\left(\Delta Y^j_r\right)^+.
	\end{split}
\end{equation*}
Thus, we have
\begin{equation*}
	\begin{split}
	&~(\Delta Y^i_r)^+\left(f^i(r,Y^{L,i}_r,Y^{L,-i}_r,\Psi^{L,i}_r)-f^i(r,\overline Y^{i}_r,\overline Y^{-i}_r,\overline \Psi^{i}_r)\right) \\
	 \leq&~ -\int_{\cE}(\Delta Y^i_r)^+\Delta\Psi^i_r(e)\mathbf{1}_{\{\Delta Y^i_r+\Delta\Psi^i_r(e)\leq0\}}\nu(\de)+(\Delta Y^i_r)^+\sum_{j\neq i}q^{ij}(\Delta Y^j_r)^+.
	\end{split}
\end{equation*}
Taking the above inequality into \eqref{estimate4}, we get
\begin{align*}
	&~e^{as}\left|(\Delta Y^i_s)^+\right|^2+a\E\left[ \left. \int_s^t e^{ar}\left|(\Delta Y^i_r)^+\right|^2\dr \;\right|\mathcal G_s \right] \\
	\leq&~\E\left[ \left. e^{at}|(\Delta Y^i_t)^+|^2+2\int_s^t e^{ar}(\Delta Y^i_r)^+\sum_{j\neq i}q^{ij}(\Delta Y^j_r)^+\dr \;\right|\mathcal G_s \right] \\
	&~ -\E\bigg[ \int_s^t\int_{\cE}e^{ar}\Big( \left((\Delta Y^i_{r}+\Delta\Psi^i_r(e))^+\right)^2-((\Delta Y^i_{r})^+)^2\\
	&\qquad\qquad\qquad\qquad-2\left(1-\mathbf{1}_{ \{\Delta Y^i_r+\Delta\Psi^i_r(e)\leq0\}}\right)(\Delta Y^i_{r})^+\Delta\Psi^i_r(e)\Big)\nu(\de)\dr \;\bigg|\;\mathcal G_s \bigg].
\end{align*}
Using the elementary inequality
$
\big((x+y)^+\big)^2-2cx^+y\geq 0
$
for any $(x,y)\in\R\times\R$ and any $c\in[0,1]$, we have
\begin{equation*}\label{estimate5}
	\begin{split}
	&~e^{as}\left|(\Delta Y^i_s)^+\right|^2+a\E\left[ \left. \int_s^t e^{ar}\left|(\Delta Y^i_r)^+\right|^2\dr \;\right|\mathcal G_s \right] \\
	\leq&~ \E\left[ \left. e^{at}\left|(\Delta Y^i_t)^+\right|^2+2\int_s^t e^{ar}(\Delta Y^i_r)^+\sum_{j\neq i}q^{ij}(\Delta Y^j_r)^+\dr \;\right|\mathcal G_s \right] +\nu(\cE)\E\left[\left.\int_s^te^{ar}|(\Delta Y^i_{r})^+|^2\dr\;\right|\mathcal G_s\right] \\
	\leq&~ \E\left[\left.e^{at}\left|(\Delta Y^i_t)^+\right|^2+\breve c\sum_{j=1}^{\ell}\int_s^te^{ar}\left|(\Delta Y^j_{r})^+\right|^2\dr \;\right|\mathcal G_s \right],
	\end{split}
\end{equation*}
where $\breve c\geq \ell \max_{i\neq j}q^{ij}+\nu(\cE)$. Taking sums in terms of $i$ from $1$ to $\ell$, we get
\begin{equation*}
	\begin{split}
	&~\sum_{j=1}^{\ell}e^{as}\left|(\Delta Y^j_s)^+\right|^2+a\E\left[\left.\sum_{j=1}^{\ell}\int_s^t e^{ar}\left|(\Delta Y^j_r)^+\right|^2\dr\;\right|\mathcal G_s\right]\\
	\leq&~ \sum_{j=1}^{\ell}\E\left[\left.e^{at}\left|(\Delta Y^j_t)^+\right|^2\;\right|\mathcal G_s\right]+\breve c\ell\mathbb E\left[ \left. \sum_{j=1}^{\ell}\int_s^te^{ar}\left|(\Delta Y^j_{r})^+\right|^2\dr \;\right|\mathcal G_s \right].
	\end{split}
\end{equation*}
By choosing $a=\breve c\ell$, we deduce from the above inequality that
\begin{align*}
	\sum_{j=1}^{\ell}e^{as}\left|( Y^{L,j}_s-\overline Y^j_s)^+\right|^2=\sum_{j=1}^{\ell}e^{as}\left|(\Delta Y^j_s)^+\right|^2
	&\leq \sum_{j=1}^{\ell}\E\left[ \left. e^{at} \left|(\Delta Y^j_t)^+\right|^2 \;\right|\mathcal G_s \right]\\
	&=\sum_{j=1}^{\ell}\E\left[ \left. e^{at}\left|(Y^{L,j}_t-\overline Y^j_t)^+\right|^2 \;\right|\mathcal G_s \right].
\end{align*}
Noting that $\overline Y^j$ is bounded from below by $S$, we have
$$
\left|(Y^{L,j}_t-\overline Y^j_t)^+\right|^2\leq \sup_{0\leq t\leq T} \left|(Y^{L,j}_t -S_t\wedge n )^+\right|^2,
$$
which is conditionally integrable for each $n$. By the dominated convergence theorem, recalling that $Y^{L,j}\in S^{\infty}_{\mathbb{G}}(\R)$ by Lemma \ref{lem:bound-Y} and that $\lim_{t\nearrow T}\overline Y_t^j=\infty$, we get
\begin{align*}
	\sum_{j=1}^{\ell}e^{as}\left|( Y^{L,j}_s-\overline Y^j_s)^+\right|^2\leq \lim_{t\nearrow T}\sum_{j=1}^{\ell}\E\left[ \left. e^{at} \left|(Y^{L,j}_t-\overline Y^j_t \wedge n )^+\right|^2 \;\right|\mathcal G_s \right] = \sum_{j=1}^\ell e^{aT}\left|(L-n )^+\right|^2 .
\end{align*}
Letting $n\rightarrow\infty$, this implies $Y^{L,j}_s\leq\overline Y^j_s$ a.s. for all $s\in[0,T)$ and $j\in\cM$. Finally, sending $L\rightarrow\infty$ yields the claim. 
\end{proof}

\section{Solvability of the control problem \eqref{cost}-\eqref{state}\\
and uniqueness of the BSDE system \eqref{BSDE} }\label{sec:wellposedness-control}

In this section, we will solve the state constrained stochastic control problem \eqref{cost}-\eqref{state} by considering the limit in Section \ref{sec:unconstrained}.
As a byproduct, we will verify that the minimal solution of the BSDE system \eqref{BSDE} established in Theorem \ref{thm:existence} is the unique one.
%
%
%

\subsection{Solve the control problem \eqref{cost}-\eqref{state}}

Define the value function as
\begin{align*}
	V_t(x,i)&:= \inf_{(\xi,\beta)\in\mathcal A_0}J_t(\xi,\beta;x,i)\\
	& :=\inf_{(\xi,\beta)\in\mathcal A_0}\mathbb E\left[\left. \int_t^T 		 \left(\eta_s^{\alpha_{s-}} \xi_s^2+\lambda_s^{\alpha_{s-}} X_s^2
	+\int_{\cE}\gamma_s^{\alpha_{s-}}(e) \beta^2_s(e)\,\nu(\de)\right)\ds \;\right|\mathcal F_t\right],
\end{align*}
where $(X_t,\alpha_t)=(x,i)$ is the initial state and initial regime.

\begin{theorem}\label{thm:verification}
	Let Assumption \ref{ass1} hold. Let $(Y^i,Z^i,\Psi^i)_{i\in\cM}$ be the minimal solution of \eqref{BSDE} established in Theorem \ref{thm:existence}. Then the optimal liquidation rates in the traditional venue and in the dark pool venue in feedback form are
	\begin{align}\label{optimal-control}
		\hat \xi^{}_s= \frac{Y^{\alpha_{s-}}_{s}}{\eta^{\alpha_{s-}}_s}X_{s}, \qquad
		\hat \beta^{}_s(e)= \frac{Y^{\alpha_{s-}}_{s-}+\Psi^{\alpha_{s-}}_s(e)}
		{\gamma^{\alpha_{s-}}_s(e)+Y^{\alpha_{s-}}_{s-}+\Psi^{\alpha_{s-}}_s(e)}X_{s-}.
	\end{align}
The optimal position admits the following expression
\begin{equation}\label{optimal-position}
	\hat X^{}_s=x\exp\left(-\int_t^s\frac{Y^{\alpha_{r-}}_{r}}{\eta^{\alpha_{r-}}_r}\dr\right)
	\prod_{t<r\leq s}\left(1-\int_{\cE}\frac{Y^{\alpha_{r-}}_{r-}+\Psi^{\alpha_{r-}}_r(e)}
	{\gamma^{\alpha_{r-}}_r(e)+Y^{\alpha_{r-}}_{r-}+\Psi^{\alpha_{r-}}_r(e)}N(\{r\},\de)\right).
\end{equation}
In particular, $\hat X$ is nonincreasing and $0\leq \hat X_s\leq x$.
Moreover, the value function starting at $(t,X_t,\alpha_t)=(t,x,i)$ satisfies
	\begin{align}\label{value}
		V_t(x,i)= Y^{i}_{t}x^2.
	\end{align}
\end{theorem}
\begin{proof}
Substituting \eqref{optimal-control} into the state process \eqref{state}, we have
\begin{equation*}
	\left\{\begin{split}
		\dd X_s=&~-X_{s-}\left(\frac{Y^{\alpha_{s-}}_{s}}{\eta^{\alpha_{s-}}_s}\ds
+\int_{\cE}\frac{Y^{\alpha_{s-}}_{s-}+\Psi^{\alpha_{s-}}_s(e)}
		{\gamma^{\alpha_{s-}}_s(e)+Y^{\alpha_{s-}}_{s-}+\Psi^{\alpha_{s-}}_s(e)} \,N(\ds,\de)\right),\\
		X_{t}=&~x, \quad \alpha_{t}=i,
	\end{split}\right.
\end{equation*}
whose unique solution is given by \eqref{optimal-position}. To verify $\hat X$ is the optimal position, it remains to verify that the trading rates in \eqref{optimal-control} are admissible and optimal.

{\bf Step 1: $(\hat\xi,\hat\beta)$ is admissible.} We use It\^o's formula (refer to e.g. \cite[Lemma 4.3]{Hu2020} for the differential w.r.t. a Markov chain)
to $Y^{\alpha_{s}}_s \hat X^2_s$ from $s=t$ to $s=\hat t\in (t,T)$, and obtain
\begin{equation}\label{veri}
	\begin{split}
	&~ Y^{\alpha_{\hat t}}_{\hat t} \hat X_{\hat t}^2+\int_t^{\hat t} \left(\eta_s^{\alpha_{s-}}\hat \xi_s^2+\lambda_s^{\alpha_{s-}}\hat X_s^2
	+\int_{\cE}\gamma_s^{\alpha_{s-}}(e) \hat \beta^2_s(e)\,\nu(\de)\right)\ds \\
	=&~ Y^{\alpha_{\hat t}}_{\hat t}\hat X_{\hat t}^2+\int_t^{\hat t} \hat X_{s-}^2\left(\eta_s^{\alpha_{s-}}\left|\frac{Y^{\alpha_{s-}}_{s-}}{\eta^{\alpha_{s-}}_s}\right|^2
	+\lambda_s^{\alpha_{s-}}
	+\int_{\cE}\gamma_s^{\alpha_{s-}}(e)\left|\frac{Y^{\alpha_{s-}}_{s-}+\Psi^{\alpha_{s-}}_s(e)}
	{\gamma^{\alpha_{s-}}_s(e)+Y^{\alpha_{s-}}_{s-}+\Psi^{\alpha_{s-}}_s(e)}\right|^2\nu(\de)\right)\ds\\
	=&~Y^{\alpha_t}_t\hat X^2_t +\int_t^{\hat t} \hat X_s^2(Z^{\alpha_{s-}}_s)^{\top}\dw_s + \int_t^{\hat t}\hat X_{s-}^2\sum_{j,j'\in\cM}(Y^j_s-Y^{j'}_s)\mathbf{1}_{\{\alpha_{s-}=j'\}} \dd\widetilde N^{j'j}_s \\
	&~+\int_t^{\hat t} \hat X_{s-}^2 \int_{\cE} \left((Y^{\alpha_{s-}}_{s-}+\Psi^{\alpha_{s-}}_s(e))\left(\left(1-\frac{Y^{\alpha_{s-}}_{s-}+\Psi^{\alpha_{s-}}_s(e)}	{\gamma^{\alpha_{s-}}_s(e)+Y^{\alpha_{s-}}_{s-}+\Psi^{\alpha_{s-}}_s(e)}\right)^2-1\right)+\Psi^{\alpha_{s-}}_s(e)\right)\widetilde N(\ds,\de),
	\end{split}
\end{equation}
where $\{N^{j'j}\}_{j,j'\in\cM}$ are independent Poisson processes with intensities $q^{j'j}$, and $\widetilde N^{j'j}_t=N^{j'j}_t-q^{j'j}t$ are the corresponding compensated Poisson processes. It follows that the process
\begin{align*}
	\zeta^t_{\hat t}:=Y^{\alpha_{\hat t}}_{\hat t} \hat X_{\hat t}^2 - Y^{\alpha_t}_t \hat X^2_t + \int_t^{\hat t} \left(\eta_s^{\alpha_{s-}} \hat\xi_s^2+\lambda_s^{\alpha_{s-}}\hat X_s^2
	+\int_{\cE}\gamma_s^{\alpha_{s-}}(e)|\hat\beta^{\alpha_{s-}}_s(e)|^2\,\nu(\de)\right)\ds
\end{align*}
is a nonnegative martingale on $[t,T)$. Thus, it converges almost surely as $\hat t\nearrow T$. Since $\{Y^i\}_{i\in\cM}$ is uniformly positive and satisfies the terminal condition $\lim_{\hat t\nearrow T}Y_{\hat t}^i=+\infty, \ i\in\cM$, we obtain that
\begin{align*}
	0\leq \hat X_{\hat t}^{}\leq \sqrt{ \frac{ \zeta^t_{\hat t} + Y^{\alpha_t}_t \hat X^2_t }{Y^{\alpha_{\hat t}}_{\hat t}}}\rightarrow 0,\quad \textrm{ a.s. as }\hat t\nearrow T.
\end{align*}
As a result, $\hat X_{T-}=0$ and thus $\hat X_T=0$.

Letting $t=0$ and taking expectations on both sides of \eqref{veri}, we have
\begin{equation}\label{eq:xi-L2}
	\begin{split}
		\infty>Y^{i_0}_{0}x^2_0=&~\E\left[Y^{\alpha_{\hat t}}_{\hat t}\hat X_{\hat t}^2\right]+\E\left[\int_0^{\hat t} \left(\eta_s^{\alpha_{s-}}\hat \xi_s^2+\lambda_s^{\alpha_{s-}}\hat X_s^2
		+\int_{\cE}\gamma_s^{\alpha_{s-}}(e)|\hat \beta_s(e)|^2\,\nu(\de)\right)\ds\right]\\
		\geq&~\E\left[\int_0^{\hat t} \left(\eta_s^{\alpha_{s-}}\hat \xi_s^2+\lambda_s^{\alpha_{s-}}\hat X_s^2
		+\int_{\cE}\gamma_s^{\alpha_{s-}}(e)|\hat \beta_s(e)|^2\,\nu(\de)\right)\ds\right].
	\end{split}
\end{equation}
Thus,
$(\hat \xi,\hat\beta)\in\mathcal{A}_0$.


{\bf Step 2: $(\hat\xi,\hat\beta)$ is optimal.}
On the one hand,
taking conditional expectations on both sides of \eqref{veri}, we have for $t<T$
\begin{equation*}
	\begin{split}
	 Y^{i}_{t} x^2 =&~\E\left[\left.Y^{\alpha_{\hat t}}_{\hat t}\hat X_{\hat t}^2 \;\right|\mathcal\; F_t \right]+\E\left[ \left. \int_t^{\hat t} \left(\eta_s^{\alpha_{s-}}\hat \xi_s^2+\lambda_s^{\alpha_{s-}}\hat X_s^2
	+\int_{\cE}\gamma_s^{\alpha_{s-}}(e)|\hat \beta_s(e)|^2\,\nu(\de)\right)\ds \;\right|\mathcal F_t \right]\\
	\geq&~\E\left[ \left. \int_t^{\hat t} \left(\eta_s^{\alpha_{s-}}\hat \xi_s^2+\lambda_s^{\alpha_{s-}}\hat X_s^2
	+\int_{\cE}\gamma_s^{\alpha_{s-}}(e)|\hat \beta_s(e)|^2\,\nu(\de)\right)\ds \;\right|\mathcal F_t \right].
	\end{split}
\end{equation*}
Taking the limit $\hat t\nearrow T$ and using the monotone convergence theorem, we have
\begin{equation}\label{geq}
	\begin{split}
	Y^{i}_{t} x^2 \geq &~ \E\left[ \left. \int_t^T \left(\eta_s^{\alpha_{s-}}\hat \xi_s^2+\lambda_s^{\alpha_{s-}}\hat X_s^2
	+\int_{\cE}\gamma_s^{\alpha_{s-}}(e)|\hat \beta_s(e)|^2\,\nu(\de)\right)\ds \;\right|\mathcal F_t \right]\\
	=&~J_t(\hat\xi,
	\hat \beta;x, i)\geq V_t(x,i).
	\end{split}
\end{equation}
On the other hand, for any $(\xi,\beta)\in\mathcal{A}_0$, we have $J_t(\xi,\beta;x,i)= J^L_t(\xi,\beta;x,i)$, where $J_t^L$ was defined in \eqref{cost-L}. It implies that for each $L\geq 0$
\begin{equation}
	\begin{split}
	V_t(x,i)=\inf_{(\xi,\beta)\in\mathcal{A}_0}J_t(\xi,\beta;x,i)
	= &~\inf_{(\xi,\beta)\in\mathcal{A}_0}J^L_t(\xi,\beta;x,i)\\
	\geq&~ \inf_{(\xi,\beta)\in L^2_{\mathbb F}(\mathbb R)\times L^2_{\mathcal P^{\mathbb F}}(\mathbb R)}J^L_t(\xi,\beta;x,i)=V^L_t(x,i)=Y^{L,i}_{t}x^2.
	\end{split}
\end{equation}
By \eqref{sup-convergence-Y} we obtain
$
	Y^{i}_{t}x^2 =\lim_{L\nearrow\infty}Y^{L,i}_{t}x^2 \leq V_t(x,i)
$,
which together with \eqref{geq} implies \eqref{value} as well as the optimality of $(\hat \xi,\hat\beta)$.
\end{proof}

\subsection{The uniqueness result of the BSDE system \eqref{BSDE}}
\begin{theorem}\label{thm:uniqueness}
	The minimal solution established in Theorem \ref{thm:existence} is unique in the sense of Definition \ref{def:solution} such that $Y^i_t>0$ for each $i\in\mathcal M$ and $t\in[0,T]$.
\end{theorem}
\begin{proof}
	Let $(Y^i,Z^i,\Psi^i)_{i\in\mathcal M}$ be any solution of the BSDE system \eqref{BSDE}. Our goal is to prove that $Y^i$ coincides with the value function $V_t(1,i)$. By \eqref{value}, uniqueness follows.
	
	{\bf Step 1.} Estimate for $Y^i$. Let $\delta>0$ and define
	\[
		Y^\delta_t = \frac{ \eta^\star }{T-\delta-t } + \frac{\lambda^\star(T-\delta-t)}{3}, \quad 0\leq t < T-\delta,
	\]
where we recall $\eta^\star$ and $\lambda^\star$ are the constants in Assumption \ref{ass1}. It follows that
\[
	-\dd Y^\delta_t = \left( \lambda^\star - \frac{2Y^\delta_t}{T-\delta -t} + \frac{\eta^\star}{(T-\delta-t)^2} \right)\dt - Z^\delta_t\dw_t,\quad \lim_{t\nearrow T-\delta}Y^{\delta}_t=+\infty,
\]
where $Z^\delta_t\equiv0$.
By Proposition \ref{prop:comparison}, we have $Y^i_t\leq Y^\delta_t$ for each $t\in[0,T-\delta)$.
Letting $\delta\downarrow 0$, we have $Y^i_t\leq \lim_{\delta\downarrow 0} Y^{\delta}_t = \frac{\eta^\star}{T-t}+\frac{\lambda^\star(T-t)}{3}$, for $0\leq t<T$.

In the subsequent {\bf Step 2} and {\bf Step 3}, w.l.o.g., we restrict our attention to the admissible space
$$
\mathcal A_{0,b}:=\left\{ (\xi,\beta)\in L^2_\mathbb{F}(\mathbb{R})\times L^{2}_{\mathcal{P}^{\mathbb F}}(\mathbb{R}):~X^{\xi,\beta}_T=0,~\beta\textrm{ is bounded and }\lim_{s\nearrow T}\beta_s=0 \right\},
$$
noting $(\hat\xi,\hat\beta)\in \mathcal A_{0,b}$ from \eqref{optimal-control} and \eqref{eq:xi-L2}.

{\bf Step 2.} In this step, we prove that $\lim_{\hat t\nearrow T}\mathbb E\left[\left.Y^{\alpha_{\hat t}}_{\hat t} \left(X^{\xi,\beta}_{\hat t}\right)^2 \;\right|\mathcal F_t \right]=0$, for any $(\xi,\beta)\in\mathcal A_{0,b}$.
By the liquidation constraint, $X^{\xi,\beta}$ can be rewritten as a backward equation
\[
	 X^{\xi,\beta}_{\hat t} = \int_{\hat t}^T \xi_s\ds + \int_{\hat t}^T\int_{\mathcal E} \beta_s(e) N(\ds,\de):=\int_{\hat t}^T \xi_s\ds + \widetilde X_{\hat t}^{\beta}.
\]
By It\^o's formula, we have $(\widetilde X^\beta_t)^2=\int_t^T\int_{\mathcal E} (2\widetilde X^\beta_{s-}\beta_s(e)-\beta^2_s(e))\,N(\ds,\de) $.
By {\bf Step 1}, we have
\begin{equation*}
	\begin{split}
		0\leq&~ \mathbb E\left[ \left. Y^{\alpha_{\hat t}}_{\hat t} \left(X^{\xi,\beta}_{\hat t}\right)^2 \;\right|\mathcal F_t \right] \\
		 \leq&~ \frac{C}{T-\hat t}\mathbb E\left[\left.\left(\int_{\hat t}^T \xi_s\ds \right)^2 \;\right|\mathcal F_t \right] + \frac{C}{T-\hat t}\mathbb E\left[\left. \int_{\hat t}^T\int_{\mathcal E} (2\widetilde X^\beta_{s-}\beta_s(e)- \beta^2_s(e))\,N(\ds,\de) \;\right|\mathcal F_t \right] +O(T-\hat t)\\
		\leq&~ C\mathbb E\left[ \left. \int_{\hat t}^T \xi^2_s\ds \;\right|\mathcal F_t \right] + \frac{C}{T-\hat t}\mathbb E\left[\left. \int_{\hat t}^T\int_{\mathcal E} (2\widetilde X^\beta_{s-}\beta_s(e)- \beta^2_s(e))\,\nu(\de)\ds \;\right|\mathcal F_t \right] +O(T-\hat t) \\
	\overset{\hat t\nearrow T}{\longrightarrow}&~ 	C\mathbb E\left[ \left. \int_{\mathcal E}\lim_{s\nearrow T} (2\widetilde X^\beta_{s-}\beta_s(e)-\beta^2_s(e))\,\nu(\de) \;\right|\mathcal F_{t} \right]\\
		=&~0.
	\end{split}
\end{equation*}

{\bf Step 3.}
For each $i\in\mathcal M$ and let $\alpha_t=i$. For any $(\xi,\beta)\in\mathcal A_{0,b}$, applying It\^o's formula to $Y^\alpha (X^{\xi,\beta})^2$ from $t$ to $\hat t>t$, we have
\begin{equation*}
	\begin{split}
	&~ Y^{\alpha_{\hat t}}_{\hat t} \left(X^{\xi,\beta}_{\hat t}\right)^2\\
	=&~Y^{\alpha_t}_t X^2_{t}+\int_{t}^{\hat t} Y^{\alpha_{s-}}_{s}\left\{ 2X^{\xi,\beta}_{s}\left(-\xi_s-\int_{\mathcal E}\beta_s(e)\,\nu(\de)\right) + \int_{\mathcal E}\beta^2_s(e)\,\nu(\de) \right\}\ds\\
	&~-\int_t^{\hat t}\left(X^{\xi,\beta}_{s}\right)^2 f^{\alpha_{s-}}(s, Y^{\alpha_s}_s,\Psi^{\alpha_s}_s )\ds + \int_t^{\hat t}\int_{\mathcal E} \Psi^{\alpha_{s-}}_s(e) \left( \beta^2_s(e)-2X^{\xi,\beta}_{s-}\beta_s(e) \right)\nu(\de)\ds \\
	&~+ \int_t^{\hat t} \left(X^{\xi,\beta}_{s}\right)^2 \sum_{k\in\mathcal M} q^{\alpha_{s-}k}Y^k_s\ds+\int_t^{\hat t} (Y^{\alpha_{s-}}_{s-}+\Psi^{\alpha_{s-}}_s(e) ) \int_{\mathcal E} \left( \beta^2_s(e) - 2X^{\xi,\beta}_{s-}\beta_{s}(e) \right)\widetilde N(\ds,\de) \\
	&~+ \int_t^{\hat t} \left(X^{\xi,\beta}_{s-}\right)^2\int_{\mathcal E}\Psi^{\alpha_{s-}}_s(e)\widetilde N(\ds,\de)\\
	&~+\int_t^{\hat t}\left(X^{\xi,\beta}_{s}\right)^2 Z^{\alpha_{s-}}_s\dw_s+\int_t^{\hat t}\left( X^{\xi,\beta}_{s-}\right)^2 \sum_{k,k'\in\mathcal M} (Y^k_s-Y^{k'}_s) \mathbf{1}_{ \{ \alpha_{s-}=k' \} }\dd\widetilde N^{k'k}_s,
	\end{split}
\end{equation*}
which implies that by taking the expression of $f^{\alpha}$ \eqref{def:f} into account
\begin{align*}
		&~	Y^{\alpha_{\hat t}}_{\hat t} \left(X^{\xi,\beta}_{\hat t}\right)^2 + \int_t^{\hat t} \left(\eta^{\alpha_{s-}}_s \xi^2_s + \lambda^{\alpha_{s-}}_s \left(X^{\xi,\beta}_s\right)^2 + \int_{\mathcal E} \gamma^{\alpha_{s-}}_s(e) \beta^2_s(e)\,\nu(\de) \right)\ds\\
		=&~Y^{i}_t X^2_{t} + \int_t^{\hat t} \left( -2X^{\xi,\beta}_{s} Y_{s}^{\alpha_{s-}}\xi_s +\eta^{\alpha_{s-}}_s \xi^2_s + \frac{\left(X^{\xi,\beta}_{s}\right)^2(Y^{\alpha_{s-}}_{s})^2}{
		\eta_s^{\alpha_{s-}}} \right)\ds \\ 											
		&~+\int_{t}^{\hat t}\int_{\mathcal E} \bigg\{ \frac{\left(X^{\xi,\beta}_{s}\right)^2(Y^{\alpha_{s-}}_{s}+\Psi^{\alpha_{s-}}_s(e))^2}{\gamma^{\alpha_{s-}}_s(e)+Y^{\alpha_{s-}}_{s}+\Psi^{\alpha_{s-}}_s(e)} - 2X^{\xi,\beta}_{s}(Y^{\alpha_{s-}}_{s}+\Psi^{\alpha_{s-}}_s) \beta_s(e) \\
		&~\qquad\qquad\qquad+ \left( \gamma^{\alpha_{s-}}_s(e)+ Y^{\alpha_{s-}}_{s}+\Psi^{\alpha_{s-}}_s(e) \right)\beta^2_s(e) \bigg\} \nu(\de)\ds\\
		&~+\int_t^{\hat t} (Y^{\alpha_{s-}}_{s-}+\Psi^{\alpha_{s-}}_s(e)) \int_{\mathcal E} \left( \beta^2_s(e) - 2X^{\xi,\beta}_{s-}\beta_{s}(e) \right)\widetilde N(\ds,\de) + \int_t^{\hat t} \left(X^{\xi,\beta}_{s-}\right)^2\int_{\mathcal E}\Psi^{\alpha_{s-}}_s(e)\widetilde N(\ds,\de)\\
		&~+\int_t^{\hat t}\left(X^{\xi,\beta}_{s}\right)^2 Z^{\alpha_{s-}}_s\dw_s+\int_t^{\hat t} \left(X^{\xi,\beta}_{s-}\right)^2 \sum_{k,k'\in\mathcal M} (Y^k_s-Y^{k'}_s) \mathbf{1}_{ \{ \alpha_{s-}=k' \} }\dd\widetilde N^{k'k}_s\\
		=&~Y^{i}_t X^2_{t} + \int_t^{\hat t} \eta^{\alpha_{s-}}_s\left( \xi_s-\frac{Y^{\alpha_{s-}}_{s}X^{\xi,\beta}_{s}}{\eta_s} \right)^2\ds \\
		&~+ 	 \int_{t}^{\hat t}\int_{\mathcal E} \left( \gamma^{\alpha_{s-}}_s(e)+ Y^{\alpha_{s-}}_{s-}+\Psi^{\alpha_{s-}}_s(e) \right) \left(\beta_s(e) - \frac{ ( Y^{\alpha_{s-}}_{s-} + \Psi^{\alpha_{s-}}_s(e) )X^{\xi,\beta}_{s-} }{ \gamma^{\alpha_{s-}}_s(e)+ Y^{\alpha_{s-}}_{s-}+\Psi^{\alpha_{s-}}_s(e) } \right)^2 						 \nu(\de)\ds\\
		&~+\int_t^{\hat t} \left(Y^{\alpha_{s-}}_{s-}+\Psi^{\alpha_{s-}}_s(e)\right) \int_{\mathcal E} \left( \beta^2_s(e) - 2X^{\xi,\beta}_{s-}\beta_{s}(e) \right)\widetilde N(\ds,\de) + \int_t^{\hat t} X^2_{s-}\int_{\mathcal E}\Psi^{\alpha_{s-}}_s(e)\widetilde N(\ds,\de)\\
		&~+\int_t^{\hat t}X^2_{s} Z^{\alpha_{s-}}_s\dw_s+\int_t^{\hat t} \left(X^{\xi,\beta}_{s-}\right)^2 \sum_{k,k'\in\mathcal M} (Y^k_s-Y^{k'}_s) \mathbf{1}_{ \{ \alpha_{s-}=k' \} }\dd\widetilde N^{k'k}_s.
\end{align*}
Taking conditional expectations, we get
\begin{equation*}
	\begin{split}
		Y^{i}_t X^2_t \leq \mathbb E\left[ \left. Y^{\alpha_{\hat t}}_{\hat t} \left(X^{\xi,\beta}_{\hat t}\right)^2 + \int_t^{\hat t} \left(\eta^{\alpha_{s-}}_s \xi^2_s + \lambda^{\alpha_{s-}}_s \left(X^{\xi,\beta}_s\right)^2 + \int_{\mathcal E} \gamma^{\alpha_{s-}}_s(e) \beta^2_s(e)\,\nu(\de) \right)\ds \;\right|\mathcal F_t \right].
	\end{split}
\end{equation*}
Letting $\hat t\nearrow T$, by {\bf Step 2}, it holds
\[
		Y^{i}_t X^2_t \leq \mathbb E\left[ \left. \int_t^{T} \left(\eta^{\alpha_{s-}}_s \xi^2_s + \lambda^{\alpha_{s-}}_s \left(X^{\xi,\beta}_s\right)^2 + \int_{\mathcal E} \gamma^{\alpha_{s-}}_s(e) \beta^2_s(e)\,\nu(\de) \right)\ds \;\right|\mathcal F_t \right].
\]
By the arbitrariness of $(\xi,\beta)$, we have $Y^{i}_t X^2_t\leq V_t(X_t,i)$. In particular, $Y^i_t\leq V_t(1,i)$.
Thus, $Y^i$ is the minimal solution of \eqref{BSDE} by \eqref{value} and the uniqueness result follows.
\end{proof}

\section{Conclusion}
In this paper, we examined a stochastic control problem with regime switching, which arises in the context of an optimal liquidation problem. To address this, we solved a system of BSDEs with singular terminal values, marking the first instance of such a solution in the literature. For the existence result, we initially considered a truncated version of the BSDE system and then allowed the truncation to approach infinity. To establish a limit, we employed a multidimensional comparison theorem. Compared to the existence result, our uniqueness result-which is novel even in the absence of regime switching-holds greater significance for the literature. To demonstrate uniqueness, we proved that the solution obtained by taking the limit is a minimal solution, and through a verification argument, we showed that all other solutions are no larger than this minimal solution.

\appendix

\section{Proof of Proposition \ref{prop:comparison}}
In the proof, $C$ is a positive constant that may vary from line to line.

Let $\Delta Y^i=Y^i-\overline Y^i$, $\Delta Z^i=Z^i-\overline Z^i$ and $\Delta\Psi^i=\Psi^i-\overline\Psi^i$.
By It\^o's formula, we have by Condition $1$
\begin{equation}\label{eq:proof-comparison-1}
	\begin{split}
		&~[ ( \Delta Y^i_t)^+]^2 \\
		=&~ \int_t^T 2( \Delta Y^i_{s})^+\left( f^i(s,Y_{s},\Psi^i_s) - \overline f^i(s,\overline Y_{s}, \overline\Psi^i_s) \right) \ds - \int_t^T 1_{ \{ \Delta Y^i_{s}>0 \} } | \Delta Z^i_s |^2\ds\\
		&~-\int_t^T\int_{\mathcal E} \Big( [( \Delta Y^i_{s}+\Delta\Psi^i_s(e) )^+]^2 - [( \Delta Y^i_{s} )^+]^2 -2(\Delta Y^i_{s})^+\Delta\Psi^i_s(e) \Big)\nu(\de)\ds \\
		&~-\int_t^T 2(\Delta Y^i_{s-})^+(\Delta Z^i_s)^\top\dw_s + \int_t^T\int_{\mathcal E} \Big( [(\Delta Y^i_{s-}+\Delta\Psi^i_s(e))^+]^2 - [ (\Delta Y^i_{s-})^+ ]^2 \Big)\widetilde N(\ds,\de).
	\end{split}
\end{equation}
First, we consider the estimate of the first term on the right hand side. By the assumptions on the drivers, it holds that
\begin{equation*}
	\begin{split}
		&~ f^i(s,Y_{s},\Psi^i_s) - \overline f^i(s,\overline Y_{s}, \overline\Psi^i_s) \leq 	 f^i(s,Y_{s},\Psi^i_s) - f^i(s,\overline Y_{s}, \overline\Psi^i_s)\qquad (\textrm{by Condition 5})\\
		=&~\Big( f^i(s,Y_{s},\Psi^i_s) - f^i(s,Y_{s},\overline\Psi^i_s) \Big) + \Big( f^i(s,Y_{s},\overline\Psi^i_s) -f^i(s,\overline Y^i_{s}, Y^{-i}_{s},\overline\Psi^i_s) \Big) + \Big( f^i(s,\overline Y^i_{s},Y^{-i}_{s}, \overline \Psi^i_s) - f^i(s,\overline Y_{s}, \overline\Psi^i_s) \Big)\\
		\leq&~ \int_{\mathcal E}| \Delta\Psi^i_s(e) |\nu(\de)	\qquad 	\qquad (\textrm{by Condition }4)	 \\
		&~+ \Big( f^i(s,Y_{s},\overline\Psi^i_s) -f^i(s,\overline Y^i_{s}, Y^{-i}_{s},\overline\Psi^i_s) \Big) \\
		&~+ C\sum_{j\neq i}( \Delta Y^j_{s} )^+ \quad \qquad \qquad (\textrm{by Condition }2).
	\end{split}
\end{equation*}
Thus, the first term on the right hand side of \eqref{eq:proof-comparison-1} can be estimated as follows
\begin{equation*}
	\begin{split}
		&~\int_t^T 2(\Delta Y^i_{s})^+ \left( f^i(s,Y_{s},\Psi^i_s) - \overline f^i(s,\overline Y_{s}, \overline\Psi^i_s) \right) \ds \\
		\leq&~ 2\int_t^T(\Delta Y^i_{s})^+\int_{\mathcal E}| \Delta \Psi^i_s(e)|\, \nu(\de)\ds + C\int_t^T ((\Delta Y^i_s)^+)^2 \ds+ C\int_t^T (\Delta Y^i_{s})^+\sum_{j\neq i}(\Delta Y^j_s)^+ \ds\qquad (\textrm{by Condition }3).
	\end{split}
\end{equation*}

Taking the above estimate into \eqref{eq:proof-comparison-1} and taking conditional expectations, we have
\begin{align*}
		[(\Delta Y^i_t)^+]^2 \leq&~ C\mathbb E\left[\left. 	\int_t^T ((\Delta Y^i_s)^+)^2 \ds	+ \int_t^T (\Delta Y^i_{s})^+\sum_{j\neq i}(\Delta Y^j_s)^+ \ds 	 \;\right|\mathcal G_t \right]\\
		&~- \mathbb E\left[\left. \int_t^T\int_{\mathcal E} \Big( [( \Delta Y^i_{s}+\Delta\Psi^i_s(e) )^+]^2 - [( \Delta Y^i_{s} )^+]^2 -2 (\Delta Y^i_{s})^+\Delta\Psi^i_s(e) -2 (\Delta Y^i_{s})^+|\Psi^i_s(e)| \Big)\,\nu(\de)\ds \;\right|\mathcal G_t \right]\\
		\leq&~ C\sum_{i=1}^\ell\mathbb E\left[\left. \int_t^T ( (\Delta Y^i_s)^+ )^2\ds \;\right|\mathcal G_t \right] + C \mathbb E\left[\left. \int_t^T (( \Delta Y^i_s )^+)^2\ds \;\right|\mathcal G_t \right] \\
		\leq&~ C\sum_{i=1}^\ell\mathbb E\left[\left. \int_t^T ( (\Delta Y^i_s)^+ )^2\ds \;\right|\mathcal G_t \right].
\end{align*}
Let $\epsilon=\frac{1}{2 C\ell}$ and define $\mathcal Y(\epsilon)=\esssup_{\omega\in\Omega,s\in[T-\epsilon,T],i\in\mathcal M}| (\Delta Y^i_s)^+ |^2$. The above implies that $\mathcal Y(\epsilon)\leq C\ell \epsilon\mathcal Y(\epsilon)=\frac{1}{2}\mathcal Y(\epsilon)$. Thus, $(\Delta Y^i)^+\equiv 0$ on $[T-\epsilon,T]$ for all $i\in\mathcal M$. In particular, $(\Delta Y^i_{T-\epsilon})^+=0$. Starting from $T-\epsilon$ and repeating the above analysis for finitely many times, we get $(\Delta Y^i)^+\equiv 0$ on $[0,T]$, completing the proof.

\bibliography{Fu}

\end{document}